\documentclass[letterpaper, 10 pt, conference]{ieeeconf} 
\IEEEoverridecommandlockouts
\usepackage{epstopdf}
\usepackage[OT1]{fontenc} 
\usepackage{fontenc}
\IEEEoverridecommandlockouts                              
\usepackage{array}
\usepackage{tikz}
\usetikzlibrary{positioning}
\usepackage{tabularx}
\usepackage{xcolor}
\usepackage{amsmath}
\usepackage{amsfonts}
\usepackage{mathtools}
\usepackage{stmaryrd}
\usepackage[ruled,vlined]{algorithm2e}
\usepackage{tcolorbox}

\usepackage{wrapfig}
\usepackage{tikz}
\usepackage{graphicx}
\usepackage{svg}
\usepackage{subcaption}
\usepackage{float}
\usepackage{amssymb}
\usepackage[colorlinks=true, urlcolor=blue, linkcolor=red]{hyperref}
\usepackage{moresize}
\usepackage[utf8]{inputenc}
\usepackage{etoolbox}
\usepackage{lipsum}
\usepackage{subcaption}
\usepackage{newtxmath}
\let\labelindent\relax
\usepackage{enumitem}
\usepackage{scalerel}
\usepackage{epsfig}
\usepackage{cite}
\usepackage{subcaption}
\usepackage{hyperref}

\hypersetup{
    colorlinks=true,
    linkcolor=black,
    urlcolor=blue,
    pdftitle={},
    pdfpagemode=FullScreen,
    }
\usepackage{anyfontsize}

\def\PS#1{{\textcolor{orange}{ {\bf PS:} #1}}}
\def\BS#1{{\textcolor{red}{ {\bf BS:} #1}}}
\def\PJ#1{{\textcolor{green}{ {\bf PJ:} #1}}}

\usetikzlibrary{shapes.geometric, arrows}
\usetikzlibrary{calc,fit,arrows}
\newtheorem{theorem}{Theorem}[section]
\newtheorem{lemma}[theorem]{Lemma}

\newtheorem{problem}[theorem]{Problem}

\newtheorem{definition}[theorem]{Definition}

\newtheorem{remark}[theorem]{Remark}
\newtheorem{assumption}{Assumption}
\usepackage{accents}

\newlist{pcases}{enumerate}{1}
\setlist[pcases]{
  label=\underline{Case~\arabic*:}\protect\thiscase.~,
  ref=\arabic*,
  align=left,
  labelsep=0pt,
  leftmargin=0pt,
  labelwidth=0pt,
  parsep=0pt
}
\newcommand{\case}[1][]{%
  \if\relax\detokenize{#1}\relax
    \def\thiscase{}%
  \else
    \def\thiscase{~#1}%
  \fi
  \item
}

  \usepackage{cite}
\usepackage{amsmath,amssymb,amsfonts}
\usepackage{algorithmic}
\usepackage{textcomp}
\usepackage{xcolor}
\def\BibTeX{{\rm B\kern-.05em{\sc i\kern-.025em b}\kern-.08em
    T\kern-.1667em\lower.7ex\hbox{E}\kern-.125emX}}
\begin{document}
\def\PS#1{{\textcolor{orange}{ {\bf PS:} #1}}}
\def\BS#1{{\textcolor{red}{ {\bf BS:} #1}}}
\def\PJ#1{{\textcolor{green}{ {\bf PJ:} #1}}}
\title{
Sliding Mode Control for Safe Trajectory Tracking with Moving Obstacles Avoidance: Experimental Validation on Planar Robots  
\thanks{* The authors contributed equally.}
\thanks{This work was supported in part by the ANRF Advanced Research Grant and the Siemens fellowship.}
\thanks{$^1$ Centre for Cyber-Physical Systems and $^2$ Department of Aerospace Engineering, Indian Institute of Science, Bengaluru, India. 
{Emails: \{shubhamsg, sangeerthp, saharsh2021, \newline pushpak\}@iisc.ac.in}}
\thanks{Link:https://youtu.be/dWcxwum96vk}
}

\author{Shubham Sawarkar$^{1*}$, P Sangeerth$^{1*}$, S. Saharsh$^1$, and Pushpak Jagtap$^{1,2}$}

\maketitle

\begin{abstract}
This paper presents a unified control framework for robust trajectory tracking and moving obstacle avoidance applicable to a broad class of mobile robots. By formulating a generalized kinematic transformation, we convert diverse vehicle dynamics into a strict feedback form, facilitating the design of a Sliding Mode Control (SMC) strategy for precise and robust reference tracking. To ensure operational safety in dynamic environments, the tracking controller is integrated with a Collision Cone Control Barrier Function (C3BF) based safety filter. The proposed architecture guarantees asymptotic tracking in the presence of external disturbances while strictly enforcing collision avoidance constraints. The novelty of this work lies in designing a sliding mode controller for ground robots like the Ackermann drive, which has not been done before. The efficacy and versatility of the approach are validated through numerical simulations and extensive real-world experiments on three distinct platforms: an Ackermann-steered vehicle, a differential drive robot, and a quadrotor drone. Video of the experiments are available at \href{https://youtu.be/dWcxwum96vk}{[video link]}.
\end{abstract}


\section{Introduction}
Autonomous systems operating in unstructured environments such as industrial automation, service robotics, and self-driving vehicles-require control strategies that reconcile accurate trajectory tracking with verifiable safety. As these systems increasingly interact with humans and dynamic agents, strict reference tracking and guaranteed collision avoidance under uncertainty become essential. This paper addresses these challenges by proposing an integrated control architecture that combines the robustness of Sliding Mode Control (SMC) for tracking with the safety guarantees of Control Barrier Functions (CBFs) for obstacle avoidance.

Sliding Mode Control (SMC) is renowned for its insensitivity to matched uncertainties, ensuring precise trajectory tracking despite perturbations. Advanced formulations, such as backstepping \cite{huang2024backstepping} and integral SMC \cite{ibrahim2016wheeled,li2025integral}, have demonstrated high accuracy and convergence. Practically, these benefits have been validated through hardware experiments on non-holonomic mobile platforms \cite{yang1999sliding,mu2017nonlinear} and line-following robots \cite{yildiz2020sliding}.

Safety, particularly collision avoidance, is paramount for robots operating in shared, dynamic environments. Control Barrier Functions (CBFs) address this challenge by enforcing forward-invariant safe sets \cite{ames2019control}, with collision cone concepts specifically enabling moving obstacle avoidance \cite{goswami2024collision}. These methods have been successfully applied to UGVs \cite{goswami2024collision} and UAVs \cite{C3BF}. Prior works, including leader-follower tracking \cite{zhao2023safe} and unified CLF-CBF stability frameworks \cite{ding2023guaranteed}, integrate SMC with CBFs, combining robustness with active safety. Additionally, linear SMC combined with High-Order CBFs has been validated on hardware for static obstacle avoidance \cite{andrabi2025trajectory}.

Despite extensive work on SMC and CBFs, a gap remains between theory and practice. While simulation results are common, hardware validations of robust SMC-based tracking with moving obstacle avoidance are scarce. Accordingly, this work makes the following contributions:

\begin{enumerate}
    \item To the best of the authors’ knowledge, limited experimental validation exists for SMC derived directly from the kinematic model of an Ackermann-steered vehicle in the presence of moving obstacles. 
    
    \item Unlike prior works \cite{8796166,411340} that use specialized sliding surfaces for underactuated differential drive robots, we propose a canonical transformation that yields a strict-feedback form, removing structural constraints and enabling a broad class of SMC laws \cite{utkin2013sliding}.
    
    \item We extend the tracking framework to moving obstacle avoidance in non-holonomic systems by integrating a Collision Cone Control Barrier Function (C3BF) safety filter \cite{C3BF,goswami2024collision}, ensuring collision-free tracking incorporating robust CBF to handle the disturbances.
    
    \item Rigorous hardware experiments on both Ackermann-steered and differential drive platforms validate the unified SMC--CBF framework, demonstrating robust performance in real-world dynamic environments.
\end{enumerate}

\section{Preliminaries and Problem Formulation}
\subsection{Notations}
We denote sets of positive integers, real numbers, positive real numbers, and non-negative real numbers, respectively, by $\mathbb{N}:=\{1,2,3,\dots\}$, $\mathbb{R}$, $\mathbb{R}^+$, and $\mathbb{R}^+_0$. The symbol ${\mathbb{R}^n}$ is used to denote an $\textit{n}$-dimensional Euclidean space. A column vector $x \in \mathbb{R}^n$ is denoted by $x\hspace{-0.2em}=\hspace{-0.2em}[x_1;x_2;\cdots;x_n]$. The Euclidean norm of $\boldsymbol{x} \in \mathbb{R}^n$ is represented by $\lVert \boldsymbol{x}\rVert$. For a vector $\boldsymbol{x} \in \mathbb{R}^n$, the infinity norm $\lVert \boldsymbol{x} \rVert_{\infty}:=\max\{x_1,\ldots,x_n\}$. We use bold-faced symbols $\boldsymbol{x} \in \mathbb{R}^n$ for vectors and normal symbols for scalars $x \in \mathbb{R}$. For $\boldsymbol{x} \in \mathbb{R}^n$, the notation $|\boldsymbol{x}|$ denotes the vector whose components are the absolute values of the corresponding components of $\boldsymbol{x}$. For a vector $\boldsymbol{x} \in \mathbb{R}^n$, $\operatorname{sgn}(\cdot)$ denotes the element-wise signum function. For a matrix $A \in \mathbb{R}^{m \times n}$, $A^\top$ denotes its transpose.  A continuous function $\alpha:\mathbb{R}^+_0\rightarrow\mathbb{R}^+_0$ is class-$\mathcal{K}$ if $\alpha(0)=0$ and is strictly increasing. If $\alpha\in\mathcal{K}$, and is radially unbounded, i.e., $\alpha(r)\rightarrow\infty$ as $r\rightarrow\infty$, then $\alpha$ is a class $\mathcal{K}_{\infty}$ function. An extended class-$\mathcal{K}$ function is a continuous, strictly increasing mapping
$\alpha:(-a,b)\rightarrow[0,\infty)$ with $\alpha(0)=0$, where $a,b>0$.


\subsection{System Description}
\label{system definitions subsection}
We consider a ground robot vehicle which could be modelled in the form:
\begin{align}
        \dot{\boldsymbol{p}} &= \boldsymbol{\upsilon},  \nonumber \\
        \dot{\boldsymbol{\upsilon}} &= f_{\upsilon}(\boldsymbol{p},\boldsymbol{\upsilon}) + h_{\upsilon}(\boldsymbol{p},\boldsymbol{\upsilon})(\boldsymbol{u}+\boldsymbol{d})\label{eqn: sys},
\end{align}
where $\boldsymbol{p}\in\mathbb{R}^2$ represents the position of the robot in 2D plane, $\boldsymbol{\upsilon}\in\mathbb{R}^2$ represents the velocity of the robot in 2D plane, $\boldsymbol{u}\in \mathbb{R}^2$ is the input vector and $\boldsymbol{d}\in\mathbb{R}^2$ is matched bounded disturbance. The mappings $f_{\upsilon}:\mathbb{R}^2\times\mathbb{R}^2\to\mathbb{R}^2$ and $h_{\upsilon}:\mathbb{R}^2\times\mathbb{R}^2\to\mathbb{R}^{2\times 2}$ are locally Lipschitz continuous, ensuring the existence and uniqueness of solutions \cite{Khalil:1173048}. The matched disturbance
$\boldsymbol{d}(t)$ is assumed to be unknown but bounded, i.e.,
$\|\boldsymbol{d}(t)\|_\infty \le \bar d,
$
where $\bar d>0$ is a known constant upper bound on the matched disturbance. In this work, we make the following assumption to work on a compact state-space.

\begin{assumption}
\label{ass:hv_regular}
On an admissible set $\Omega\subset\mathbb R^{4}$,
$h_{\upsilon}(\boldsymbol x)$ is uniformly nonsingular and bounded, i.e.,
$\|h_{\upsilon}(\boldsymbol x)\|_\infty\le\bar h$
for all $\boldsymbol x=(\boldsymbol p,\boldsymbol \upsilon)\in\Omega$.
\end{assumption}

Next, we formally state the problem considered in this work.

\begin{tcolorbox}[width=\linewidth]
\begin{problem}
\label{Problem-1}
Given a robot whose dynamics is given by \eqref{eqn: sys}. The aim is to design a tracking control such that the robot follows the desired trajectory $\boldsymbol{p}_\mathrm{ref} \in \mathcal{C}^2$ in practical scenarios, which include the presence of disturbance using a sliding mode control strategy, avoiding moving obstacles, where $\mathcal{C}^2$ represents the class of twice differentiable functions.
\end{problem}
\end{tcolorbox}

 \begin{remark}
 A broad class of planar robotic systems can be modeled as \eqref{eqn: sys}, for example, a simple double integrator, Ackermann-drive vehicles, differential-drive robots \cite{C3BF} as well as aerial vehicles such as drones and fixed-wing aircraft operating in altitude hold mode \cite{beard2012small}. The general dynamics for planar robots under some minor assumptions can be transformed into the form of \eqref{eqn: sys} which is elaborated in the next section.
 \end{remark}

\section{Vehicle Models}


In this section we show that specific vehicle models can be reformulated into the structure of \eqref{eqn: sys} by defining an alternative velocity vector $\boldsymbol{\upsilon}$. We illustrate this claim in this work using an Ackermann steering model and a differential drive model.

Consider the kinematic model of an Ackermann-drive vehicle (refer Fig~\ref{fig:ackermann}):
\begin{align}\label{eqn: ackermann drive}
    \dot{\boldsymbol{p}} &= v\begin{bmatrix}
        \cos(\delta_{1} + \delta_{2}) \\
        \sin(\delta_{1} + \delta_{2})
    \end{bmatrix}, \quad \dot{\delta}_{1} = \frac{v}{l_r}\sin{\delta_{2}}, \nonumber \\
    \dot{v} &= u_{1}+d_1, \quad \dot{\delta}_{3} = u_{2}+d_2,
\end{align}
where $v$ is the linear velocity, $\delta_{1}$ is the robot's heading angle, and $l_r, l_f$ are the distances from the center of mass (CoM) to the rear and front axles, respectively, $\delta_{2}= \tan^{-1}\left(\frac{l_r}{l_r + l_f}\tan(\delta_{3})\right)$ is the slip angle, which depends on the steering angle $\delta_{3}$. The inputs are the linear acceleration $u_{1}$ and the steering angular velocity $u_{2}$.

To transform this into the proposed form, let us define the alternate velocity vector $\boldsymbol{\upsilon}$ as:
\begin{align}\label{ack1}
    \dot{\boldsymbol{p}} = \boldsymbol{\upsilon} = v\begin{bmatrix}
        \cos(\delta_{1} + \delta_{2}) \\
        \sin(\delta_{1} + \delta_{2})
    \end{bmatrix}.
\end{align}
Differentiating \eqref{ack1} with respect to time yields:
\begin{align}\label{eqn: converted ackermann drive}
     &\dot{\boldsymbol{\upsilon}}= \dot{v}\begin{bmatrix}
        \cos(\delta_{1} + \delta_{2}) \\
        \sin(\delta_{1} + \delta_{2})
    \end{bmatrix} + v\begin{bmatrix}
        -\sin(\delta_{1} + \delta_{2}) \\
        \cos(\delta_{1} + \delta_{2})
    \end{bmatrix}(\dot{\delta}_{1} + \dot{\delta}_{2}) \nonumber \\
    &= \underbrace{\frac{v^2}{l_r}\sin{\delta_{2}}\begin{bmatrix}
        -\sin(\psi) \\
        \cos(\psi)
    \end{bmatrix}}_{f_{\upsilon}} + \underbrace{\begin{bmatrix}
    \cos(\psi) & -\mathcal{K}(\delta_3) v \sin(\psi) \\
    \sin(\psi) & \mathcal{K}(\delta_3) v \cos(\psi)
\end{bmatrix}}_{h_{\upsilon}}(\boldsymbol{u}+\boldsymbol{d}),
\end{align}
where $\boldsymbol{u} = [u_{1}, u_{2}]^{\top}$, $\psi = \delta_1 + \delta_2$, and $
   \mathcal{K}(\delta_3) = \frac{\frac{l_r}{l_r + l_f}\sec^2(\delta_{3})}{1 + \left(\frac{l_r}{l_r + l_f}\right)^2\tan^2(\delta_{3})}$.
Thus, the Ackermann dynamics \eqref{eqn: ackermann drive} is transformed into the form of \eqref{eqn: sys}. The state vector $\boldsymbol{\delta}$ can be chosen as either $[\delta_1, \delta_2]^\top$ or $[\delta_1, \delta_3]^\top$, since the algebraic relation between $\delta_2$ and $\delta_3$ allows one to be determined from the other. 

Similarly, consider a differential drive robot (refer Fig~\ref{fig:unicycle}) model whose kinematics are governed by $\dot{x} = v\cos{\theta},\dot{y} = v\sin{\theta},\dot{\theta} = w$, where $x, y$ represent the robot's position in the 2D plane, and $v, w$ represent the linear and angular velocities, respectively. Following a similar procedure, the above dynamics can be transformed into the structure of \eqref{eqn: sys} as follows:

\begin{align}
\label{eqn: converted unicycle model}
    \dot{\boldsymbol{p}} = \boldsymbol{\upsilon},
    \dot{\boldsymbol{\upsilon}} = \begin{bmatrix}
        \cos(\theta) & -v \sin(\theta)\\
        \sin(\theta) & v \cos(\theta)
    \end{bmatrix}\begin{bmatrix}
        u_{1}+d_1 \\
        u_{2}+d_2
    \end{bmatrix},
\end{align}
where the state vector is $\boldsymbol{p} = [x, y]^\top$, and the time-varying parameters $\theta$ and $v$ are available (measurable via IMU or odometry). The control inputs are defined as the linear acceleration $u_1$ and the angular velocity $u_2$. One can see from \eqref{eqn: converted ackermann drive} and \eqref{eqn: converted unicycle model}, that $h_\upsilon^{-1}$ does not exist when $v=0$. This clearly gives us a restriction to use the above mentioned controller using feedback linearization methods. In order to overcome this limitation, we further make the following assumption.
\begin{figure}[!ht]
    \centering
    \begin{subfigure}[t]{0.48\linewidth}
        \centering
        \includegraphics[width=1.25\linewidth]{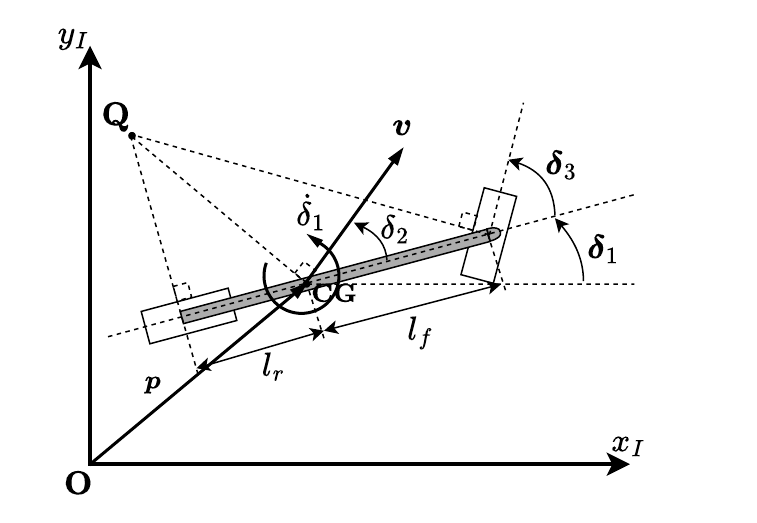}
        \caption{Ackermann Drive}
        \label{fig:ackermann}
    \end{subfigure}
    \hfill
    \begin{subfigure}[t]{0.48\linewidth}
        \centering
        \includegraphics[width=1.3\linewidth]{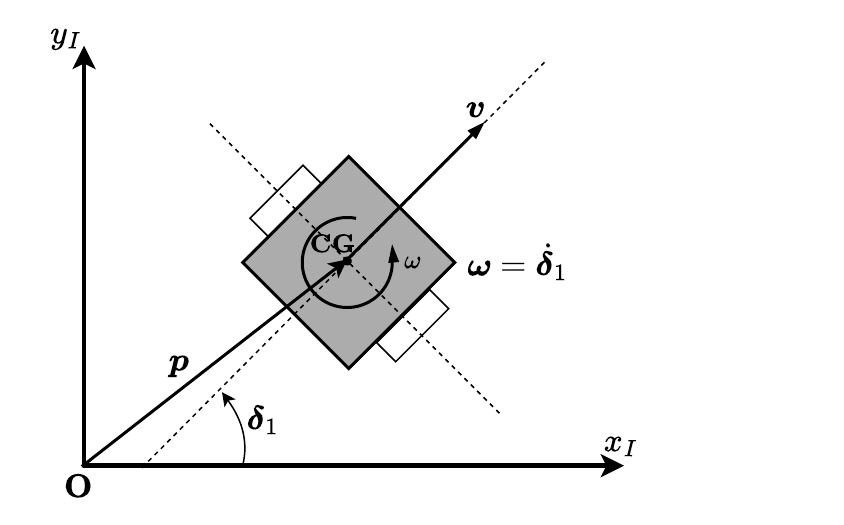}
            \caption{Differential drive}
        \label{fig:unicycle}
    \end{subfigure}
    \caption{Vehicle models}
    \label{fig:vehicle_models}
\end{figure}

\begin{assumption}
\label{ass:forward_motion}
The considered maneuvers are executed in the forward-driving mode, and the longitudinal velocity $v(t)$
is uniformly sign definite for all $t\ge0$. The reference trajectory $\boldsymbol{p}_{\mathrm{ref}}(t)$ is chosen accordingly. The maximum forward velocity $v_{\max}$ and the maximum steering angle $\delta_{3,\max}$ of the Ackermann vehicle are known and bounded.
\end{assumption}

With the above assumption, now we introduce the following lemma, which guarantees that $\lVert h_\upsilon \rVert_\infty$ is bounded for the Ackermann drive so that one can satisfy Assumption~\ref{ass:hv_regular}. One can derive a similar condition on the boundedness of $\lVert h_\upsilon \rVert_\infty$ for other mobile robots, which is not shown here for brevity. We need this boundedness in Section~\ref{sec: tracking_controller}.
\begin{lemma}
\label{lem:svd_ackermann_G}
Let $h_{\upsilon}$ be the input matrix in \eqref{eqn: converted ackermann drive}. Its singular values are $\sigma(h_{\upsilon})=\{1,|\mathcal{K}(\delta_3)v|\}$.
If $|\delta_3|\le\delta_{3,\max}<\pi/2$ and $|v|\in[v_{\min},v_{\max}]$ with $v_{\min}>0$, then $\sigma_{\min}(h_{\upsilon}) \ge \min\left\{1, \frac{l_r}{l_r+l_f}v_{\min}\right\} > 0 \text{ and}$ $\sigma_{\max}(h_{\upsilon}) 
\le \max\left\{1, \mathcal{K}(\delta_{3,\max}) v_{\max}\right\}.$ Thus, $h_{\upsilon}$ is uniformly nonsingular and bounded on this set.
\end{lemma}
\begin{proof}
    See Appendix
\end{proof}

In the subsequent sections, we present the proposed solution to the problem formulated above.



\section{Tracking Controller}
\label{sec: tracking_controller}
To track $\boldsymbol{p}_{\mathrm{ref}}$, we cast the robot dynamics into the control-affine form \eqref{eqn: sys}. This transformation enables a systematic sliding mode control design applicable to nonlinear vehicles such as Ackermann drives and Turtlebots.

\subsection{Sliding Mode Tracking Controller}
\label{sec:sliding}
Let $\boldsymbol{p}_{\mathrm{ref}}$ denote a reference trajectory. For the system \eqref{eqn: sys}, define the tracking errors
$\boldsymbol{e}_1 := \boldsymbol{p}_{\mathrm{ref}}-\boldsymbol{p}$ and
$\boldsymbol{e}_2 := \dot{\boldsymbol{p}}_{\mathrm{ref}}-\boldsymbol{\upsilon}$.

Following standard sliding mode design principles
\cite{utkin2013sliding},
we define 
a standard sliding surface of the form
$\boldsymbol{S}=\boldsymbol{s}(\boldsymbol{e}_1)+\boldsymbol{e}_2$
is adopted, where $\boldsymbol{s}(\cdot)$ may be chosen as a linear or
non-linear function of compatible sliding surface. The control input is selected as the sum of a nominal equivalent control and a robust switching term, 
$
\boldsymbol{u} = \boldsymbol{u}_{\mathrm{nom}} + \boldsymbol{u}_{\mathrm{sw}},$ with
\begin{equation}
\label{eqn:control_law_simplified}
\begin{aligned}
\boldsymbol{u}
&=\underbrace{h_{\upsilon}^{-1}\Big(\ddot{\boldsymbol{p}}_{\mathrm{ref}}- f_{\upsilon}+ \frac{\partial {s}}{\partial \boldsymbol{e}_1}\boldsymbol{e}_2\Big)}_{\boldsymbol{u}_{\mathrm{nom}}}+\underbrace{h_{\upsilon}^{-1}K\,\operatorname{sgn}(\boldsymbol{S})}_{\boldsymbol{u}_{\mathrm{sw}}}.
\end{aligned}
\end{equation}

The gain $K$ is chosen to satisfy the reaching condition
\begin{equation}
\label{eqn: bounded_k}
K >\|h_{\upsilon}\|_\infty\|\boldsymbol{d}\|_\infty+ \eta,\qquad\eta > 0.
\end{equation}

The closed-loop sliding dynamics satisfy a finite-time reaching condition, and the sliding manifold $\boldsymbol{S}=\mathbf{0}$ is reached in finite time $T \le \frac{\sqrt{\boldsymbol{S}(0)^\top \boldsymbol{S}(0)}}{\eta}$.
 For a detailed finite-time convergence analysis, the reader is referred to \cite[Ch.~3]{utkin1996integral}.

The reference trajectory is selected such that $\|\dot{\boldsymbol{p}}_{\mathrm{ref}}(t)\|$ remains sufficiently large to ensure $v(t)\ge v_{\min}$ for all $t\ge0$. To admit classical solutions, we adopt a practical sliding-mode formulation by replacing $\operatorname{sgn}(\boldsymbol{S})$ with $\mathrm{sat}(\boldsymbol{S}/\lambda)$, which ensures finite-time convergence of $\boldsymbol{S}$ to the boundary-layer manifold $\|\boldsymbol{S}\|\le\lambda$ \cite{utkin2013sliding,edwards1998sliding}.

\begin{remark}
\label{rem: h_upsilon_boundedness}
The switching gain is chosen as $K>\|h_{\upsilon}\|_{\infty}\bar d+\eta$. Using the norm inequality for a matrix $A$, $\|A\|_{\infty}\le\sqrt{m}\,\sigma_{\max}(A)$ and Lemma~\ref{lem:svd_ackermann_G}, it suffices to select $K>\sqrt{2}\,\sigma_{\max}(h_{\upsilon}(\delta_{3,\max},v_{\max}))\,\bar d+\eta$. Similarly, for differential drive $K>\sqrt{2}\,\sigma_{\max}(h_{\upsilon}(v_{\max}))\,\bar d+\eta$ is sufficient.
\end{remark}

\begin{figure}
        \centering
        \includegraphics[width=0.8\linewidth]{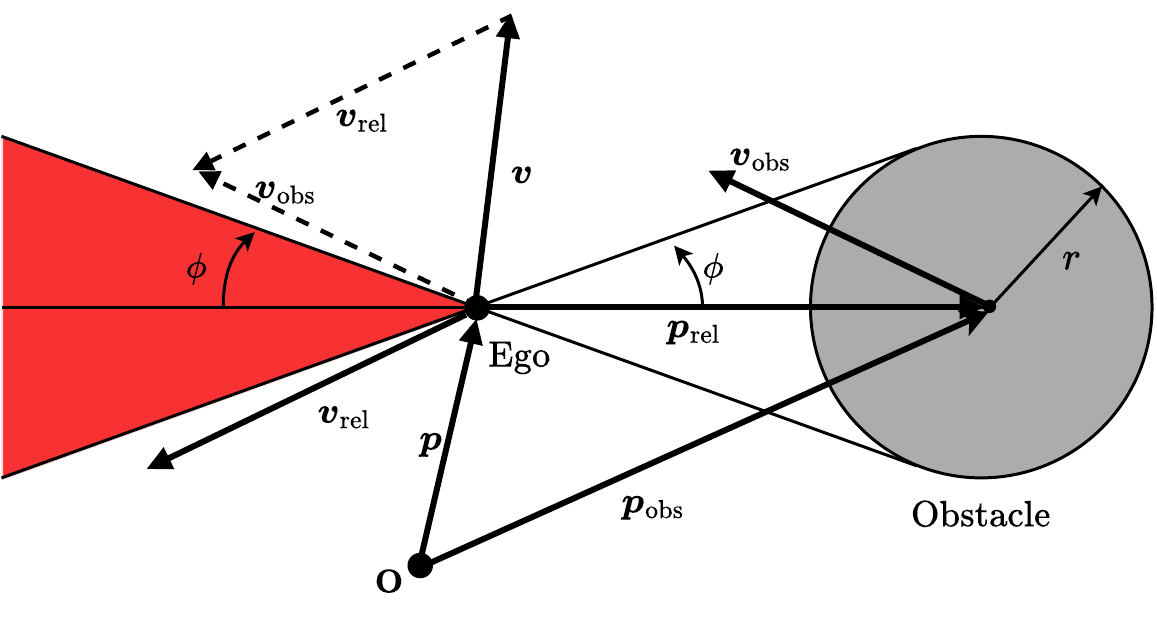}
        \caption{ Construction of collision cone for a circular obstacle considering the ego vehicle’s dimensions}
        \label{fig:cone}
\end{figure}

\section{Obstacle Avoidance}

Having introduced the vehicle dynamics and the robust tracking controller, we now present a Control Barrier Function (CBF) framework for enforcing safety constraints. We adopt the control barrier function formulation, which guarantees forward invariance of a prescribed safe set. This formulation underlies most CBF--QP safety filters used in modern safety-critical control \cite{amestheory,xu2016robustness,JANKOVIC2018359}.

Consider the matched-disturbance control-affine system
\begin{equation}
\dot{\boldsymbol{x}} = f(\boldsymbol{x}) + g(\boldsymbol{x})\big(\boldsymbol{u}+\boldsymbol{d}\big),
\qquad \|\boldsymbol{d}\|_\infty \le \bar d,\ \forall t\ge0,
\label{eq:matched_disturbed}
\end{equation}
where $\boldsymbol{x}\in\mathcal{D}\subseteq\mathbb{R}^n$, $\boldsymbol{u}\in\mathcal{U}\subseteq\mathbb{R}^m$, and $f,g$ are locally Lipschitz, ensuring the existence and uniqueness of solutions \cite{Khalil:1173048}. Now we introduce a few preliminaries on CBF.

\subsection{Robust Time-varying Control Barrier Function}
We first introduce the definition of forward invariance.
\begin{definition}[Time-varying Forward Invariance {\cite{lindemann2019stl}}]
\label{def:tv_forward_invariance_safety}
The time-varying set $\mathcal{C}(t)$ is \emph{forward invariant} for
\eqref{eq:matched_disturbed} if
$\boldsymbol{x}(0)\in\mathcal{C}(0)$ implies
$\boldsymbol{x}(t)\in\mathcal{C}(t)$ for all $t\ge 0$.
The system is said to be \emph{safe with respect to}
$\mathcal{C}(t)$ if $\mathcal{C}(t)$ is forward invariant.
\end{definition}

Now, consider a time-varying set $\mathcal{C}(t)$ defined as the super level set
of a continuously differentiable function
$h:\mathcal{D}\times\mathbb{R}_{0}^+\rightarrow\mathbb{R}$ yielding,
\begin{align}
\label{eq:tv_safe_set}
\mathcal{C}(t)
&=\{\boldsymbol{x} \in \mathcal{D} \subset \mathbb{R}^n : h(\boldsymbol{x},t)\ge 0\}, \nonumber\\
\partial \mathcal{C}(t)
&=\{\boldsymbol{x} \in \mathcal{D} \subset \mathbb{R}^n : h(\boldsymbol{x},t)=0\}, \nonumber\\
\mathrm{Int}(\mathcal{C}(t))
&=\{\boldsymbol{x} \in \mathcal{D} \subset \mathbb{R}^n : h(\boldsymbol{x},t)>0\}.
\end{align}
One can verify if the controller designed is ensuring forward invariance or not using the concept of CBFs, which is introduced next.
\begin{definition}[Time-varying Robust CBF{\cite{lindemann2019stl,jankovic2022rcbf}}]
\label{def:tv_rcbf}
Given the set $\mathcal{C}(t)$ defined by \eqref{eq:tv_safe_set}, with $L_g h(\boldsymbol{x},t)\neq 0$ for all $x\in\partial\mathcal{C}(t)$, the function $h(\boldsymbol{x},t)$ is a time-varying robust Control Barrier Function on a set $\mathcal{D}$ if there exists an extended class-$\mathcal{K}$ function $\alpha$ such that for all $(\boldsymbol{x},t)\in\mathcal{D}\times\mathbb{R}_{\ge 0}$ and all disturbances satisfying $\|\boldsymbol d(t)\|_\infty\le \bar d$,
\begin{equation*}
\sup_{\boldsymbol{u}\in\mathcal{U}}
\Big(L_t h(\boldsymbol{x},t)+ L_f h(\boldsymbol{x},t)+ L_g h(\boldsymbol{x},t)(\boldsymbol{u}+\boldsymbol{d})+ \alpha(h(\boldsymbol{x},t))\Big)\ge 0 .
\label{eq:tvr_cbf_def}
\end{equation*}
\end{definition}

Using Hölder’s inequality and the disturbance bound, the worst-case disturbance contribution satisfies
$|L_g h(\boldsymbol{x},t)\boldsymbol{d}|\le \|L_g h(\boldsymbol{x},t)\|_1\,\bar d$.
A sufficient condition for forward invariance of $\mathcal{C}(t)$ under bounded matched disturbance is therefore
\begin{equation*}
L_t h(\boldsymbol{x},t) + L_f h(\boldsymbol{x},t)+ L_g h(\boldsymbol{x},t)\boldsymbol{u}+ \alpha\!\big(h(\boldsymbol{x},t)\big)\;\hspace{-0.25em}\ge\hspace{-0.25em}\;\|L_g h(\boldsymbol{x},t)\|_1\,\bar d .
\end{equation*}
This condition defines a set of admissible control inputs for safety, from which a control action is selected via a CBF-QP.

\subsection{Time-varying CBF-QP Safety Filter}
Given a nominal locally Lipschitz control input
$\boldsymbol{u}_{\mathrm{nom}}(\boldsymbol{x},t)$, safety is enforced by solving
\begin{equation}
\label{eq:rcbf_qp_matched}
\begin{aligned}
\boldsymbol{u}^\star = \arg\min_{\boldsymbol{u}\in\mathcal{U}} \quad &
\|\boldsymbol{u}-\boldsymbol{u}_{\mathrm{nom}}(\boldsymbol{x},t)\|^2 \\
\text{s.t.}\;
L_t h(\boldsymbol{x},t)+L_f h(\boldsymbol{x},t) + & L_g h(\boldsymbol{x},t)\,\boldsymbol{u}
+ \alpha\!\big(h(\boldsymbol{x},t)\big)\\
& \ge \|L_g h(\boldsymbol{x},t)\|_1\,\bar d .
\end{aligned}
\end{equation}

The objective minimizes deviation from the nominal controller, while the constraint enforces forward invariance of $\mathcal{C}$, guaranteeing obstacle avoidance under bounded matched disturbances. Having established the general safety and robust barrier conditions, we note that distance-based control barrier functions are not well-suited for Ackerman and differential-drive dynamics, as they may lead to $L_g h=0$ and loss of control authority. We therefore adopt a collision cone-based control barrier function (C3BF)\cite{C3BF,goswami2024collision} that is compatible with such dynamics, which is introduced next.

\subsection{Collision Cone Control Barrier Function (C3BF)}

Consider a control affine system defined in \eqref{eqn: sys}. Let the relative position and relative velocity of an obstacle with respect to
the ego vehicle be defined as $\boldsymbol{p}_{\mathrm{rel}} := \boldsymbol{p}_{\mathrm{obs}}-\boldsymbol{p}$, $\boldsymbol{\upsilon}_{\mathrm{rel}} := \boldsymbol{\upsilon}_{\mathrm{obs}}-\boldsymbol{\upsilon},$
where $\boldsymbol{p}_{\mathrm{obs}},\boldsymbol{\upsilon}_{\mathrm{obs}}$ denote the obstacle position and velocity, respectively. The obstacle is conservatively approximated by a disk of radius $r>0$. The collision cone is defined as the set of relative velocity directions that would lead to an intersection with this disk. Let $\varphi\in(0,\pi/2)$ denote the half-angle of the collision cone, given by $\cos\varphi
= \frac{\sqrt{\|\boldsymbol{p}_{\mathrm{rel}}\|^2-r^2}}
{\|\boldsymbol{p}_{\mathrm{rel}}\|}$ (refer Fig~\ref{fig:cone}). Collision avoidance is enforced by ensuring that the relative velocity remains outside the collision cone (Fig. \ref{fig:cone}). Accordingly, the collision cone control barrier function (C3BF) time-varying version is defined as \cite{C3BF}
\begin{equation}
\label{eqn:c3bf}
h_{\mathrm{c3bf}}(\boldsymbol{p}_{\mathrm{rel}},\boldsymbol{\upsilon}_{\mathrm{rel}},t)
:= \boldsymbol{p}_{\mathrm{rel}}^\top \boldsymbol{\upsilon}_{\mathrm{rel}}(t) + \|\boldsymbol{p}_{\mathrm{rel}}\|\,\|\boldsymbol{\upsilon}_{\mathrm{rel}}(t)\|\cos\varphi .
\end{equation}

The associated safe set is $\mathcal{C}_{\mathrm{c3bf}}(t):=\{(\boldsymbol{p}_{\mathrm{rel}},\boldsymbol{\upsilon}_{\mathrm{rel}}(t))\mid h_{\mathrm{c3bf}}(\boldsymbol{p}_{\mathrm{rel}},\boldsymbol{\upsilon}_{\mathrm{rel}}(t))\ge 0\}$, where along the system trajectories the C3BF derivative satisfies $\dot h_{\mathrm{c3bf}} = L_f h_{\mathrm{c3bf}} + L_g h_{\mathrm{c3bf}}\,(\boldsymbol{u}+\boldsymbol{d}) + L_t h_{\mathrm{c3bf}}$, with $L_f h_{\mathrm{c3bf}}:=\nabla_x h_{\mathrm{c3bf}}^\top f(\boldsymbol{x})$, $L_g h_{\mathrm{c3bf}}:=\nabla_x h_{\mathrm{c3bf}}^\top g(\boldsymbol{x})$ for $\boldsymbol{x}=[\boldsymbol{p},\boldsymbol{\upsilon}]$ given by \eqref{eqn: sys}, and $L_t h_{\mathrm{c3bf}}:=\nabla_{\boldsymbol{\upsilon}_{\mathrm{rel}}} h_{\mathrm{c3bf}}^\top \dot{\boldsymbol{\upsilon}}_{\mathrm{obs}}(t)$ capturing the explicit time variation induced by the obstacle acceleration. Since obstacle acceleration $\boldsymbol{\dot{\upsilon}}_\mathbf{obs}$ and velocity $\boldsymbol{\upsilon}_\mathrm{obs}$ is not directly measurable, we adopt the following assumption.

\begin{assumption}
\label{ass:obstacle_motion}
Each obstacle follows either constant-velocity linear motion $\boldsymbol{\dot{\upsilon}}_\mathbf{obs}=0$ or pure circular motion with $\boldsymbol{\dot{\upsilon}}_\mathbf{obs}(t)$ depending on parameters $v_\mathrm{obs},\omega_\mathrm{obs},\boldsymbol{p}_c,$ and $\boldsymbol{p}_{\mathrm{obs}}(t)$ with known constant linear $v_\mathrm{obs}$, angular speed $\omega_\mathrm{obs}$, circle center $\boldsymbol{p}_c$ and its position $\boldsymbol{p}_{\mathrm{obs}}(t)$ is assumed known.
\end{assumption}
This assumption is reasonable for many practical scenarios, as moving obstacles in structured environments often exhibit either approximately constant-velocity motion or circular trajectories, such as pedestrians, ground vehicles, or rotating machinery. Moreover, these motion models provide sufficiently rich dynamics to capture common obstacle behaviors while maintaining analytical tractability. 

\begin{remark}
While safety is guaranteed under bounded matched disturbances, less conservative approaches such as ISSf-CBFs \cite{kolathyaames} or disturbance-observer-based CBFs \cite{wangdisturbed} may further improve performance and are left for future work.
\end{remark}

\begin{remark}[Validity of C3BF]
\label{rem:c3bf_validity}
The collision cone control barrier function (C3BF) is well defined for the
converted Ackermann and unicycle dynamics in
\eqref{eqn: converted ackermann drive} and \eqref{eqn: converted unicycle model}.
For sign definite $v(t)$, the C3BF depends on relative velocity and satisfies
$L_g h(\boldsymbol{x})\neq 0$, ensuring nondegenerate control authority.
In contrast to \cite{C3BF,goswami2024collision}, where a center-of-gravity (CG) shifting or approximating slip-angle strategy is employed to modify the system structure, the present work operates on the original higher-order dynamics without structural modification, while still guaranteeing safety with respect to moving obstacles.
\end{remark}

\subsection{Actuator and Model-Regularity Constraints (Soft Constraints)}

To preserve regularity of the input mapping $h_{\upsilon}$ in \eqref{eqn: sys} and to respect actuator limitations, admissible state and input constraints are enforced as \emph{soft} robust barrier constraints. Let $h_i:\mathbb{R}^n \to \mathbb{R}$, $i\in\mathcal I$, denote continuously differentiable state constraint functions including model-regularity and actuator constraints. For example in the Ackermann vehicle considered here, $\mathcal I=\{v_{\min},\, v_{\max},\, \delta_{3,\max}\}$.

These constraints are imposed in \eqref{eq:rcbf_qp_matched} through relaxed
robust inequalities of the form
\begin{align}
\label{eq:constraint_generic}
L_f h_i(\boldsymbol{x}) + L_g h_i(\boldsymbol{x})\,\boldsymbol{u}
+ \alpha_i\!\big(h_i(\boldsymbol{x})\big)
\;\ge\; \Delta_i - \delta_i, i\in\mathcal I,
\end{align}
where $\Delta_i\ge 0$ represents a robustness margin accounting for bounded
disturbances, and $\delta_i\ge 0$ are slack variables penalized in the QP
objective to ensure feasibility.
For Ackermann-type vehicles, $\mathcal I$ includes speed and steering bounds,
e.g.,
$h_{v_{\min}}=\boldsymbol{\upsilon}^\top\boldsymbol{\upsilon}-v_{\min}^2$,
$h_{v_{\max}}=v_{\max}^2-\boldsymbol{\upsilon}^\top\boldsymbol{\upsilon}$,
and $h_{\delta}=\delta_{3,\max}^2-\delta_3^2$.
For unicycle models, heading-rate limits are enforced through the admissible
input set $\mathcal{U}$.

\section{Safety-Critical Tracking via Sliding Mode and CBF-QP}

We now state the main result of the paper, which combines robust sliding-mode tracking with control barrier function-based safety enforcement.

\begin{theorem}[Safety-critical practical SMC tracking]
\label{thm:main_tv}
Consider a mobile robot given by \eqref{eqn: sys} with $\|\boldsymbol d\|_\infty\le\bar d$ satisfying Assumption~\ref{ass:hv_regular}. Let $\boldsymbol{p}_\mathrm{ref}$ be a reference trajectory that the system has to track and it obeys Assumption~\ref{ass:forward_motion}. We define the nominal practical sliding-mode controller for tracking $\boldsymbol{p}_\mathrm{ref}$ defined in \eqref{eqn:control_law_simplified} as $\boldsymbol{u}_{\mathrm{smc}}
=h_{\upsilon}^{-1}\!\Big(
\ddot{\boldsymbol{p}}_{\mathrm{ref}}-f_{\upsilon}
+\tfrac{\partial \boldsymbol{s}}{\partial \boldsymbol{e}_1}\boldsymbol{e}_2
+K\operatorname{sat}(\tfrac{\boldsymbol{S}}{\lambda})
\Big)$ with $K>\bar h\bar d$. Let us assume that there are moving obstacles in the considered state-space such that the obstacle satisfy Assumption~\ref{ass:obstacle_motion}. The safety-critical control law  $\boldsymbol{u}^\star$ is designed resulting from the time-varying robust C3BF \eqref{eqn:c3bf} and soft constraints in \eqref{eq:constraint_generic} using the safety filter: \begin{equation}
\label{eq:main_qp_tv}
\begin{aligned}
u^\star \in   &\arg\min_{\boldsymbol{u}\in\mathcal{U},\,\hat{\boldsymbol{\delta}}\ge 0}\;
\|\boldsymbol{u}-\boldsymbol{u}_{\mathrm{smc}}\|^2 +\rho\|\hat{\boldsymbol{\delta}}\|^2 \\
\text{\upshape s.t.}\;&
L_t h_{c3bf}+L_f h_{\mathrm{c3bf}} + L_g h_{\mathrm{c3bf}}\boldsymbol{u} +  \alpha(h_{\mathrm{c3bf}})
\ge \|L_g h_{\mathrm{c3bf}}\|_1\bar d,\\
& L_f h_i + L_g h_i\boldsymbol{u} + \alpha_i(h_i) \ge \Delta_i-\delta_i, \quad i\in\mathcal I,
\end{aligned}
\end{equation}
renders $\mathcal{C}_{\mathrm{c3bf}}(t):=\{(\boldsymbol{p}_{\mathrm{rel}},\boldsymbol{\upsilon}_{\mathrm{rel}}(t))\mid h_{\mathrm{c3bf}}(\boldsymbol{p}_{\mathrm{rel}},\boldsymbol{\upsilon}_{\mathrm{rel}}(t))\ge 0\}$ forward invariant for all $(\boldsymbol p(0) \boldsymbol{\upsilon}(0))\in \mathcal{C}_{\mathrm{c3bf}}(0)$ avoiding obstacles. Furthermore, when safety constraints are inactive, $\boldsymbol{u}^\star=\boldsymbol{u}_{\mathrm{smc}}$, ensuring $\|\boldsymbol S(t)\|\le\lambda$ is reached in finite time for practical tracking of $\boldsymbol p_{\mathrm{ref}}$.
\end{theorem}

\begin{proof}[Sketch]
Under Assumptions~\ref{ass:hv_regular}, the practical sliding-mode controller $\boldsymbol{u}_{\mathrm{smc}}$ rendering it locally Lipschitz. This implies that the QP solution $\boldsymbol{u}^\star$ inherits this regularity \cite{amestheory}, ensuring the existence and uniqueness of closed-loop solutions. This regularity allows the C3BF constraint, constructed under Assumption~\ref{ass:obstacle_motion}, to enforce the forward invariance of $\mathcal{C}_{\mathrm{c3bf}}(t)$ in the classical sense. Finally, when safety constraints are inactive, $\boldsymbol{u}^\star=\boldsymbol{u}_{\mathrm{smc}}$, and the gain condition $K>\bar h \bar d$ drives the system to the boundary layer $\|\boldsymbol{S}\|\le\lambda$ in finite time \cite{utkin2013sliding}.
\end{proof}

\section{Case Studies}
\label{sec: case study}
To demonstrate the effectiveness of the proposed tracking controller in the presence of moving obstacles, we present three distinct case studies: (i) a Differential drive robot (Turtlebot hardware), (ii) an Ackermann drive robot (F1Tenth hardware), and (iii) a drone operating in altitude hold mode (Gazebo simulation). For the Ackermann drive and differential drive, we follow assumption \ref{ass:forward_motion} and Lemma \ref{lem:svd_ackermann_G} to make the input map $h_v$ non-singular and bounded. We use ROS2 interface for the hardware experiments to obtain the obstacle and the robot's position from the PhaseSpace motion capture system. 

\begin{figure}
    \centering
    \includegraphics[width=0.9\linewidth]{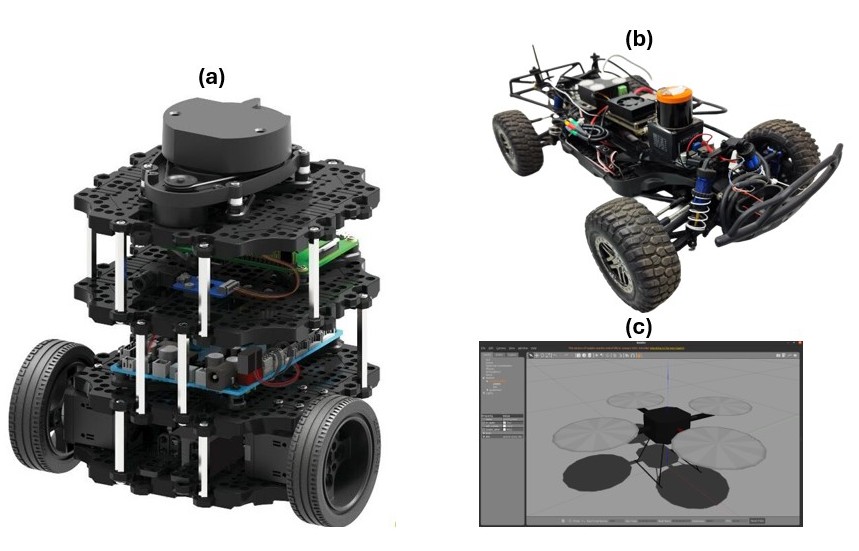}
    \caption{Robots considered in case study (a) Turtlebot (b) F1tenth car (c) Drone model in Gazebo}
    \label{fig: robots}
\end{figure}

\subsection{Ackermann Drive Vehicle (Hardware)} 
We consider an Ackermann-steered F1Tenth robot (Fig.~\ref{fig: robots} (b)) modeled by \eqref{eqn: ackermann drive} with parameters $l_f=0.17145\,\mathrm{m}$ and $l_r=0.15875\,\mathrm{m}$. A circular reference trajectory $\boldsymbol{p}_{\mathrm{ref}}=[1.58+\cos(0.2\pi t),\,1.78+\sin(0.2\pi t)]\,\mathrm{m}$ is used to track using linear SMC, with $|v|\in[0.25,3]\,\mathrm{m/s}$ and $|\delta_3|\le0.4\,\mathrm{rad}$. The dynamics are transformed into~\eqref{eqn: converted ackermann drive}.
A disturbance bound $\bar d = 0.2$ is assumed and enforced consistently in both the
sliding mode controller design and the CBF constraints.
Accordingly, the sliding gain is chosen as $K = 1$.
Using the bound $\sigma_{\max}(h_\upsilon)\approx 1.68$
(Remark~\ref{rem: h_upsilon_boundedness}), this choice satisfies the sliding
condition in \eqref{eqn: bounded_k}. Experimental results (Figs.~\ref{fig: f1 tenth xy}-\ref{fig: f1 10th sliding surface}) confirm that the controller maintains precise and robust tracking and safety throughout the maneuver demonstrated in this \href{https://youtu.be/dWcxwum96vk}{video}. In contrast to our work, \cite{LIU2020103} relies on a dynamic Ackermann model requiring difficult-to-identify inertia and friction parameters and does not address obstacle avoidance. Thus our work solves both the above problems by considering simple kinematics model.




\begin{figure}
    \centering
    \includegraphics[width=1.0\linewidth]{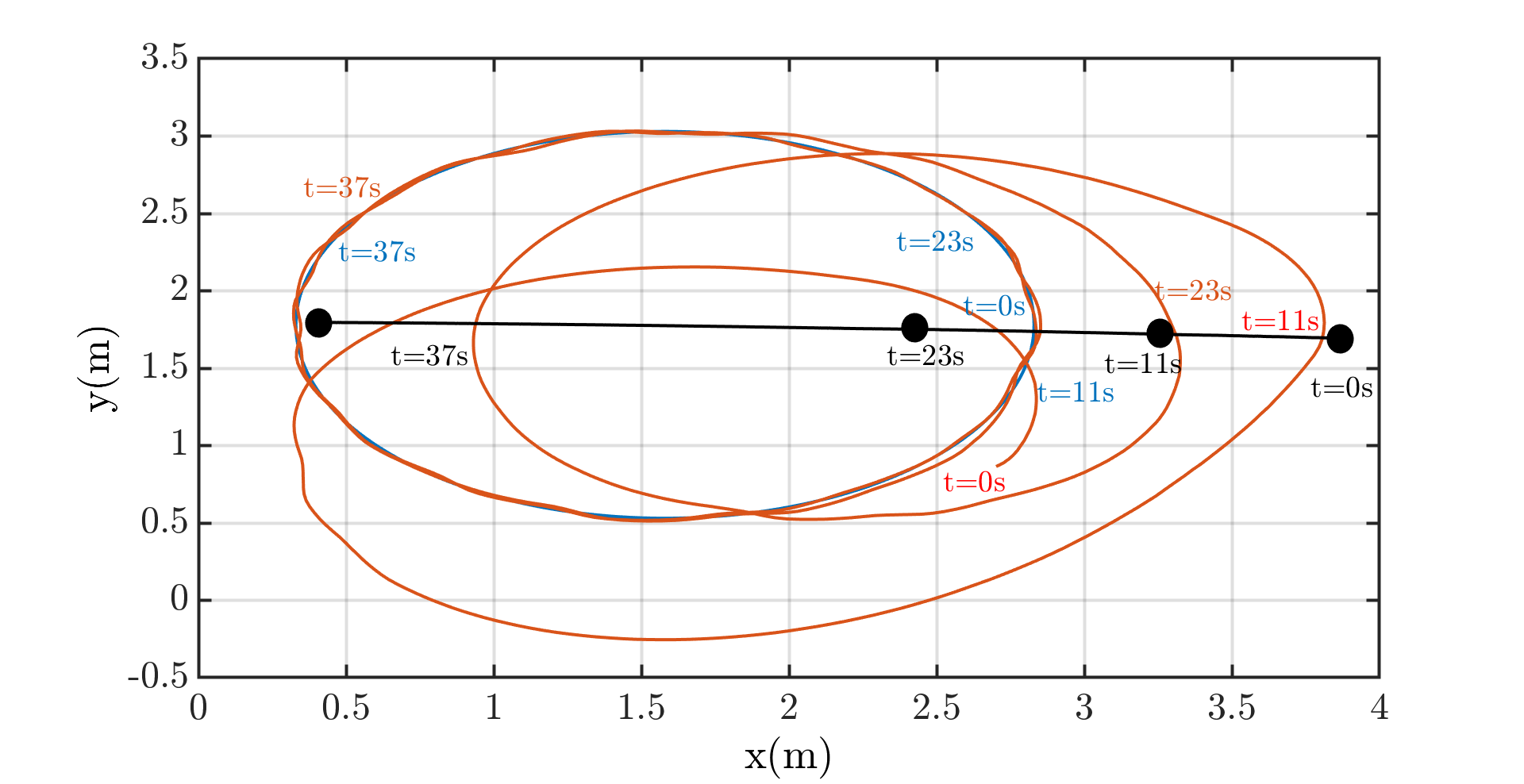}
    \caption{Ackermann drive robot performance. Reference(blue) vs actual(red) trajectory. Obstacle(black). Timestamps for the reference, robot, and obstacle are color-coded.\href{https://youtu.be/dWcxwum96vk}{[video link]}} 
    \label{fig: f1 tenth xy}
\end{figure}

\begin{figure}
    \centering
    \includegraphics[width=1.0\linewidth]{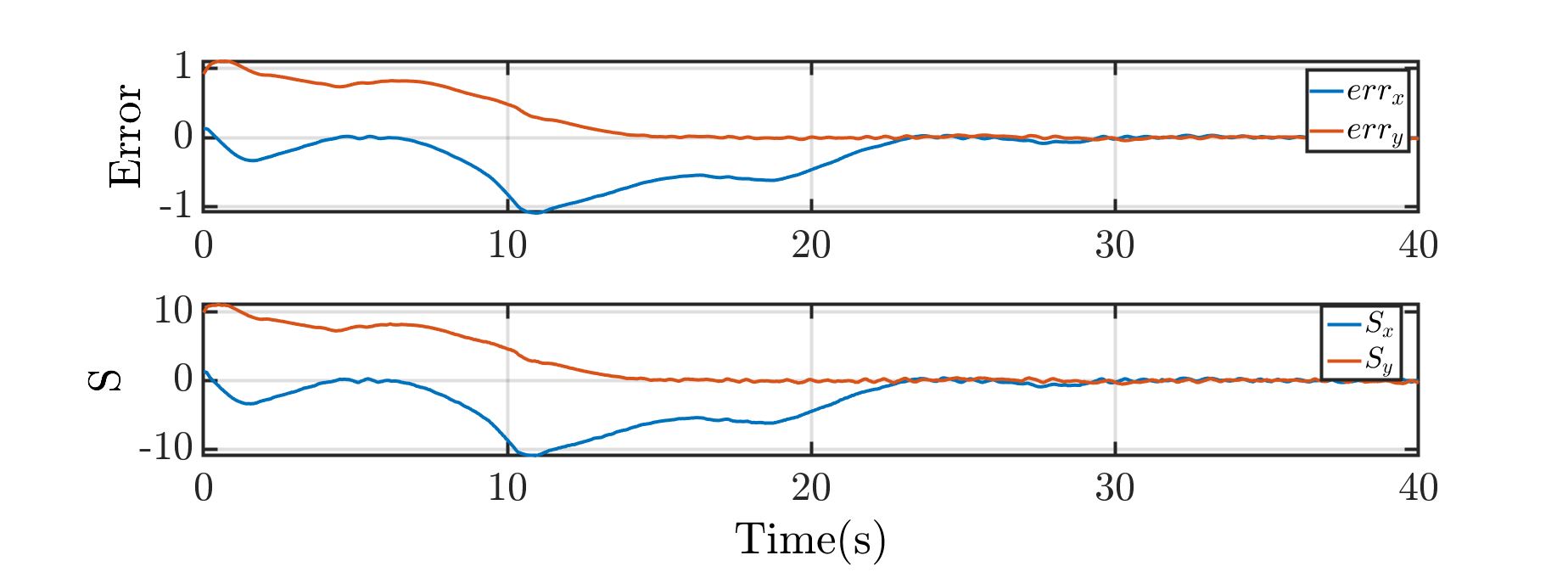}
    \caption{Ackermann drive: Time histories of sliding surface and error}
    \label{fig: f1 10th sliding surface}
\end{figure}

\subsection{Differential Drive Vehicle (Hardware)} 
We consider a differential-drive TurtleBot3 robot (Fig.~\ref{fig: robots} (a)) modeled by \eqref{eqn: converted unicycle model}. A Lissajous reference trajectory
$\boldsymbol{p}_{\mathrm{ref}}=[2.2+1.8\sin(0.23\pi t),\,1.5+0.95\cos(0.15\pi t)]m$  with $|v|\in[0.01,0.2],\mathrm{m/s}$ and $|\dot{\delta}_1|\leq2\mathrm{rad/s}$
is tracked by applying a nonsingular terminal sliding mode controller~\cite{feng2013nonsingular}. A disturbance of $\bar d=0.1$ is assumed and consistently enforced across the SMC design and the CBF constraints. Accordingly, the nominal SMC gain is chosen as $K=0.3$. Using the bound $\sigma_{\max}(h_\upsilon)\approx1$ (Remark~\ref{rem: h_upsilon_boundedness}), this choice satisfies the sliding condition in \eqref{eqn: bounded_k}. The resulting SMC reference is subsequently refined by the C3BF-QP safety filter to enforce collision avoidance against a moving obstacle following the pure circular trajectory
$\boldsymbol{p}_{\mathrm{obs}}(t)
=[2.5+0.8\cos(\omega t+\theta_0),1.75+0.8\sin(\omega t+\theta_0)\,]^\top$,
with linear speed $v_\mathrm{obs}=0.16\,\mathrm{m/s}$ and angular speed $\omega_\mathrm{obs}=0.2\,\mathrm{rad/s}$, where
$\theta_0=\tan^{-1}\big(\frac{p_{\mathrm{obs},y}(0)-1.75}{p_{\mathrm{obs},x}(0)-2.5}\big)$.
Experimental results (Figs~\ref{fig: turtlebot xy}-\ref{fig: turtlebot plots}) confirm that the controller maintains precise tracking and safety throughout the maneuver.

\begin{figure}
    \centering
    \includegraphics[width=1.0\linewidth]{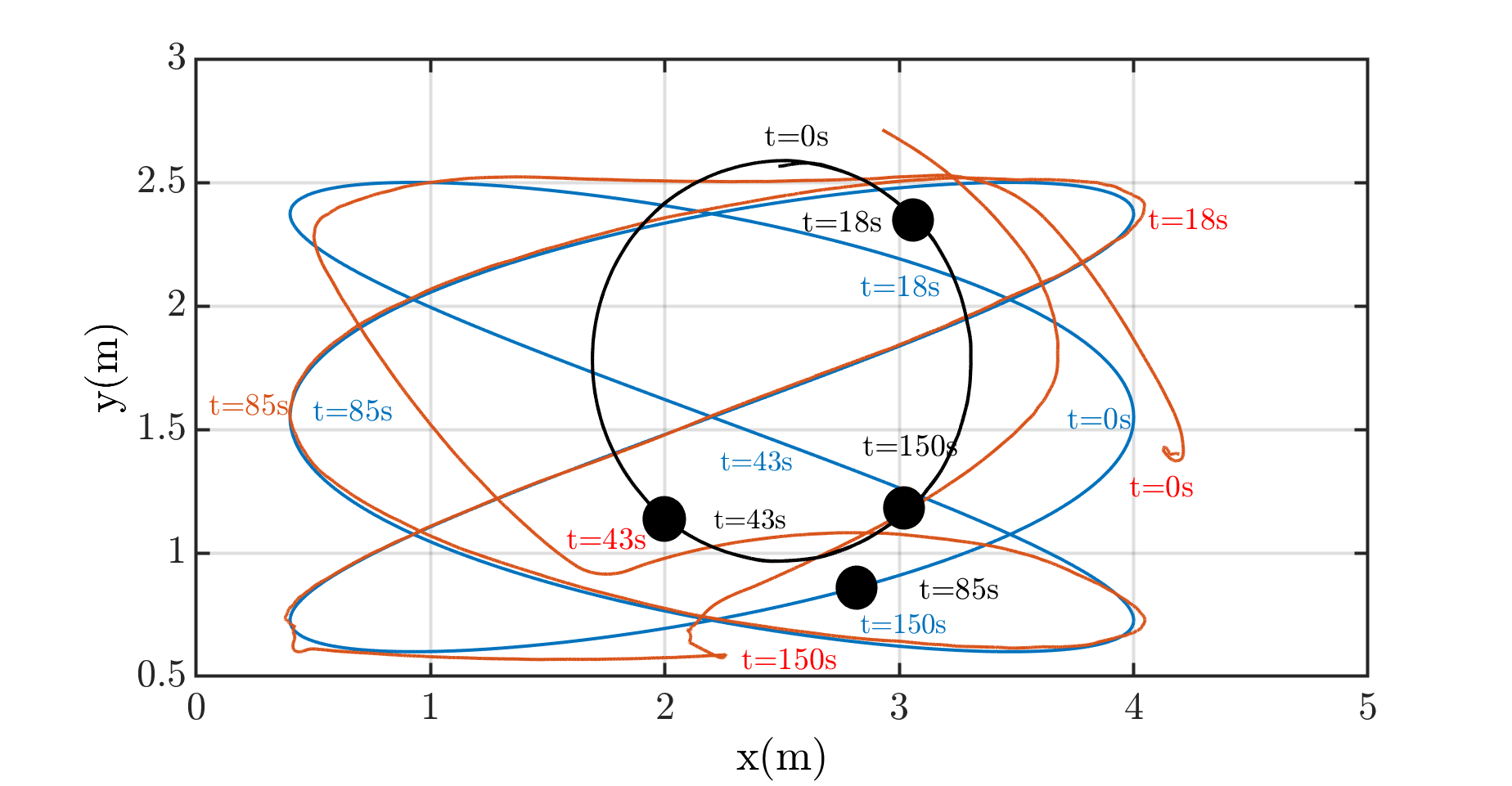}
    \caption{Differential drive performance. Curves: Reference (blue) vs. Actual (red). Region: Obstacle (black). Timestamps are color-coded to match.\href{https://youtu.be/dWcxwum96vk}{[video link]}}
    \label{fig: turtlebot xy}
\end{figure}

\begin{figure}
    \centering
    \includegraphics[width=1.0\linewidth]{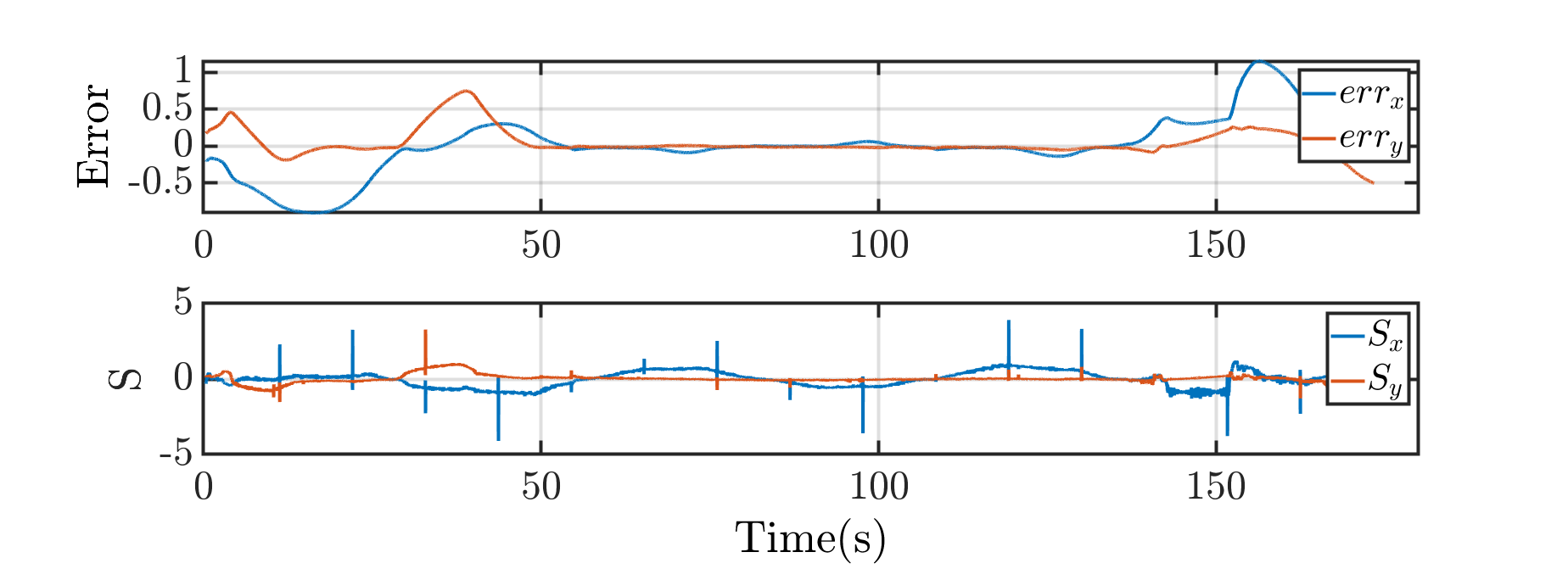}
    \caption{Differential drive: Time histories of sliding surface and error}
    \label{fig: turtlebot plots}
\end{figure}

\subsection{Drone (Simulation)} 
Finally, we evaluate the framework on a drone in a Gazebo environment (see Fig.~\ref{fig: robots} (c)) operating in altitude hold mode. The planar drone dynamics are modeled as a double integrator as done in \cite{pan2021improved}: $\dot{\boldsymbol{p}}=\boldsymbol{v}$ and $\dot{\boldsymbol{v}}=\boldsymbol{a}$, where $\boldsymbol{p}=[x, y]^\top$ is the position, $\boldsymbol{v}=[v_x, v_y]^\top$ is the velocity, and $\boldsymbol{a}=[a_x, a_y]^\top$ is the acceleration input. 
\begin{figure}
    \centering
    \includegraphics[width=1.0\linewidth]{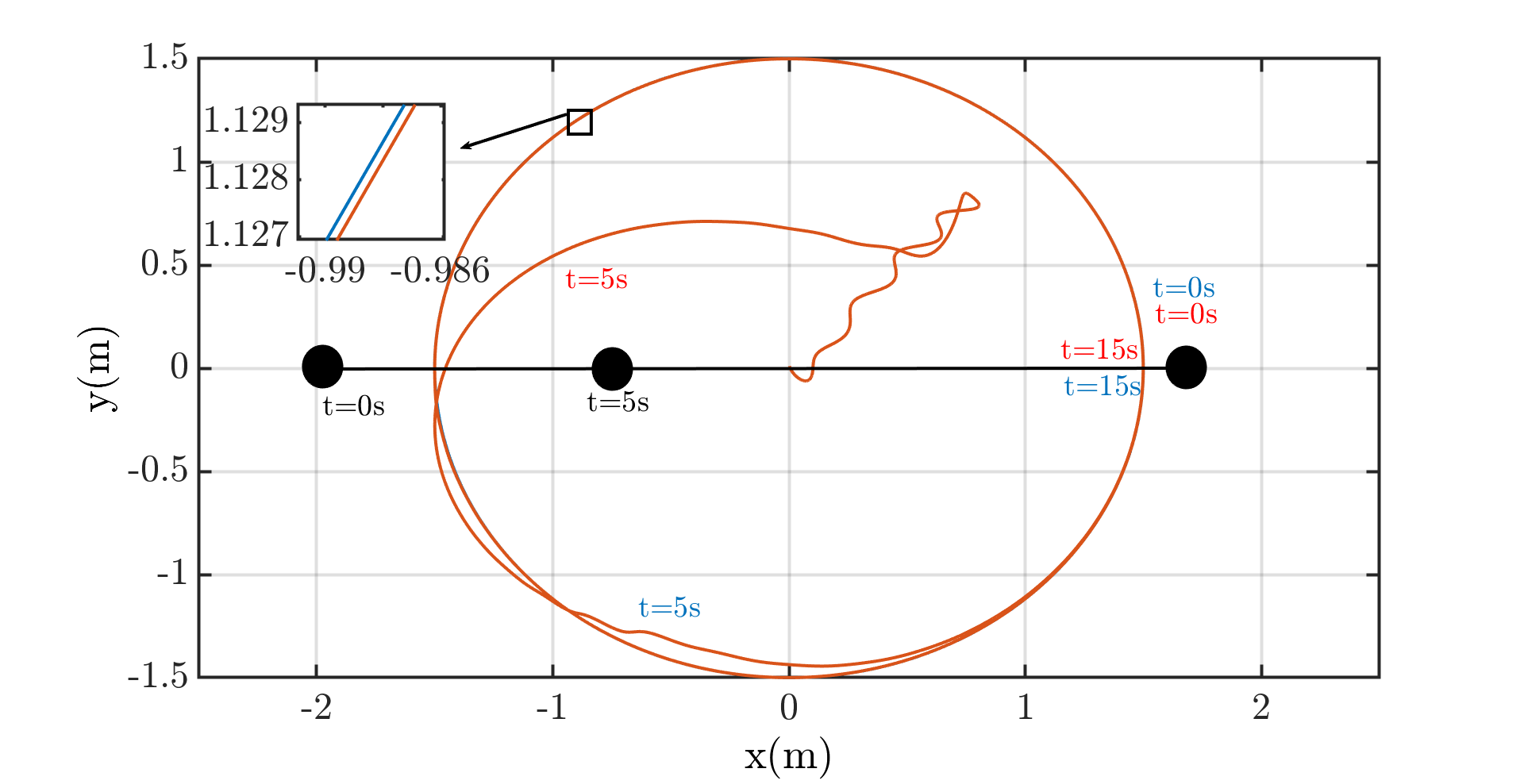}
    \caption{Drone performance in altitude hold. Curves: Reference (blue) vs. Actual (red). Region: Obstacle (black). Timestamps are color-coordinated.}
    \label{fig: drone plot xy}
\end{figure}
Since these dynamics already satisfy the form of \eqref{eqn: sys}, the sliding mode controller \eqref{eqn:control_law_simplified} is applied directly for tracking circular trajectory. Moving obstacle avoidance is enforced by incorporating the nominal input into the C3BF-QP framework, with the optimized output driving the simulated drone. The effectiveness of the approach is visualized in the supporting \href{https://youtu.be/dWcxwum96vk}{video}, with the trajectory plotted in Fig~\ref{fig: drone plot xy}.

\section{Conclusion}
We proposed a unified control strategy integrating Sliding Mode Control (SMC) and Collision Cone Control Barrier Functions (C3BF). Notably, we introduced a novel SMC formulation for Ackermann vehicles enabled by a canonical transformation. Extensive simulations and hardware validations confirmed the framework's efficacy in ensuring precise tracking and safety in dynamic environments.

\bibliographystyle{ieeetr} 
\bibliography{Sources1}

\appendix

\subsection{Proof of Lemma 3.1}
\begin{proof}
Define $R(\psi)=\begin{bmatrix}\cos\psi & -\sin\psi\\ \sin\psi & \cos\psi\end{bmatrix}$ and $D(v,\delta_3)=\mathrm{diag}(1,\mathcal{K}(\delta_3)v)$. A direct multiplication gives $h_{\upsilon}=R(\psi)D(v,\delta_3)$. Since $R(\psi)$ is orthogonal, $h_{\upsilon}^\top h_{\upsilon}=D^\top D=\mathrm{diag}(1,(\mathcal{K}(\delta_3)v)^2)$, and hence the singular values of $h_{\upsilon}$ are $1$ and $|\mathcal{K}(\delta_3)v|$.

Let $\kappa:=\frac{l_r}{l_r+l_f}\in(0,1)$ and set $t=\tan^2(\delta_3)\ge0$. Using $\sec^2(\delta_3)=1+t$, one may write $\mathcal{K}(\delta_3)=\kappa(1+t)/(1+\kappa^2 t)$, which is strictly increasing in $t$. Therefore, for all $|\delta_3|\le\delta_{3,\max}<\pi/2$, we have bounds $\kappa \le \mathcal{K}(\delta_3) \le \mathcal{K}(\delta_{3,\max})$. If $|v|\in[v_{\min},v_{\max}]$ with $v_{\min}>0$, then we see the smallest singular value is bounded below by $\sigma_{\min}(h_{\upsilon})=\min\{1,|\mathcal{K}(\delta_3)v|\} \ge\min\{1,\kappa v_{\min}\}>0$ and the largest singular value is bounded above by $\sigma_{\max}(h_{\upsilon})=\max\{1,|\mathcal{K}(\delta_3)v|\} \le\max\{1,\mathcal{K}(\delta_{3,\max})v_{\max}\}$. This proves uniform non-singularity and boundedness on the operating set.
\end{proof}

\end{document}